\documentclass[11pt,pdftex]{article}

\usepackage{amssymb}
\usepackage{amsmath}
\usepackage{epigraph}
\usepackage{paralist}
\usepackage{url}
\usepackage{tikz}
\usepackage{color} 
\usepackage{fullpage}
\usepackage{graphics}

\usepackage{microtype}

\usepackage{wrapfig}
\usepackage[mode=buildnew]{standalone}
\usepackage{tikz}
\usepackage{enumerate}
\usepackage[boxed, ruled,noend,linesnumbered]{algorithm2e}

\usepackage{caption}
\usepackage{subcaption}
\usepackage{url}

\newtheorem{theorem}{Theorem}
\newtheorem{corollary}[theorem]{Corollary}
\newtheorem{definition}{Definition}
\newtheorem{lemma}[theorem]{Lemma}
\newenvironment{proof}[1][Proof]{\noindent\textbf{#1.} }{\hfill $\Box$\\[2mm]} 

\def\A{\ensuremath{\mathcal{A}}}
\def\R{\ensuremath{\mathcal{R}}}
\def\I{\ensuremath{\mathcal{I}}}
\def\O{\ensuremath{\mathcal{O}}}
\def\Nat{\ensuremath{\mathbb{N}}}

\def\x{\mathbf{x}}

\def\Skel{\operatorname{Skel}}
\def\IS{\textit{IS}}
\def\IStwo{\ensuremath{{\textit{IS}^2}}}
\def\ISone{\ensuremath{{\textit{IS}^1}}}
\def\RK{\R_{k}}

\def\Cont{\textit{Cont}}
\def\carrier{\mathit{carrier}}

\def\s {\mathbf{s}}
\def\t {\mathbf{t}}
\def\Chr{\operatorname{Chr}}

\newcommand{\myparagraph}[1]{\vspace{2pt}\noindent \textbf{#1}}

\newcommand{\remove}[1]{}

\title{Read-Write Memory and $k$-Set Consensus as an Affine Task}

\author{
Eli Gafni
\protect\footnote{UCLA} 
\hspace{1cm}
Yuan He\protect\footnotemark[1]  
\\
\\
Petr Kuznetsov
\protect\footnote{T\'el\'ecom ParisTech} 
\hspace{1cm}
Thibault Rieutord\protect\footnotemark[2]  
}

\date{}

\begin{document}

\maketitle

\begin{abstract}
The wait-free read-write memory model has been 
characterized as an iterated \emph{Immediate Snapshot} (IS) task.
The IS task is \emph{affine}---it can be defined as a (sub)set of simplices of the
standard chromatic subdivision.
It is known that the task of \emph{Weak Symmetry Breaking} (WSB)
cannot be represented as an affine task. 
In this paper, we highlight the phenomenon of a ``natural''
model that can be captured by an iterated affine task and, thus, by a
subset of runs of the iterated immediate snapshot model.
We show that the read-write memory model in which, additionally,
$k$-set-consensus objects can be used is, unlike WSB, ``natural'' by presenting
the corresponding simple affine task captured by a subset of
$2$-round IS runs.
Our results imply the first combinatorial characterization of models
equipped with abstractions other than read-write memory that
applies to generic tasks.
\end{abstract}

\section{Introduction}

A principal challenge in distributed computing is to devise protocols
that operate correctly in the presence of failures, 
given that system components (processes) are
asynchronous.

The most extensively studied \emph{wait-free} model of computation~\cite{Her91}
makes no assumptions about the number of failures that can occur, 
no process should wait for other processes to move for
making progress.
In particular, in a \emph{wait-free} solution of a distributed \emph{task}, a process participating in
the computation should be able to produce an output regardless of the
behavior of other processes.   

\myparagraph{Topology of wait-freedom.}
Wait-free task solvability in the read-write shared-memory model has been
characterized in a very elegant way through the
existence of a specific continuous map from geometrical
structures describing inputs and outputs of the task~\cite{HS99,HKR14}. 
A task $T$ is wait-free solvable using reads and writes
if and only if there exists a simplicial, chromatic map from a
subdivision of the \emph{input} simplicial complex to the \emph{output} simplicial complex,
satisfying the specification of $T$.
Thus, using the iterated \emph{standard chromatic
  subdivision}~\cite{HKR14} (one such iteration of the
\emph{standard simplex} $\s$, denoted by $\Chr\s$, is depicted in Figure~\ref{fig:scs}), we obtain a  
combinatorial representation of the wait-free model. 
Iterations of this subdivision capture precisely rounds of the iterated immediate snapshot (IIS)
model~\cite{BG97,HS99}.  

This characterization can be interpreted as follows: the
\emph{persistent} \emph{wait-free} read-write model can be captured, regarding task solvability, by an
\emph{iterated} (one-shot) Immediate Snapshot task.
Immediate Snapshot is, in turn, 
captured by the \emph{chromatic simplex agreement}
task~\cite{BG97,HS99} on
$\Chr\s$.    

\myparagraph{Beyond wait-freedom: $k$-concurrency and $k$-set consensus.}
Unfortunately, very few  tasks are solvable
\emph{wait-free} in the read-write shared-memory model~\cite{BG93b,HS99,SZ00}, so a
lot of efforts have been applied to characterizing task solvability in 
various \emph{restrictions} of the \emph{wait-free} model.  

%
A straightforward way to define such a restriction 
is to bound  the \emph{concurrency level} of runs~\cite{GG11-univ}: 
in the model of  $k$-concurrency, at most $k$ processes can be concurrently \emph{active},
i.e., after the invocation of a task and before 
terminating it.
This is a powerful abstraction known, when it comes to solving tasks, 
to be equivalent to the wait-free
model in which processes, in addition to read-write shared memory, can
access $k$-set consensus objects~\cite{GG09}.
Also, the $k$-concurrent task solvability proved to be a good way to
measure the power of shared-memory models equipped with failure
detectors~\cite{DFGK15}. 

Can we represent the model of read-write memory and  $k$-set consensus
as an iterated task?

\begin{wrapfigure}{r}{0.47\textwidth}
\center
\includegraphics[scale=0.27]{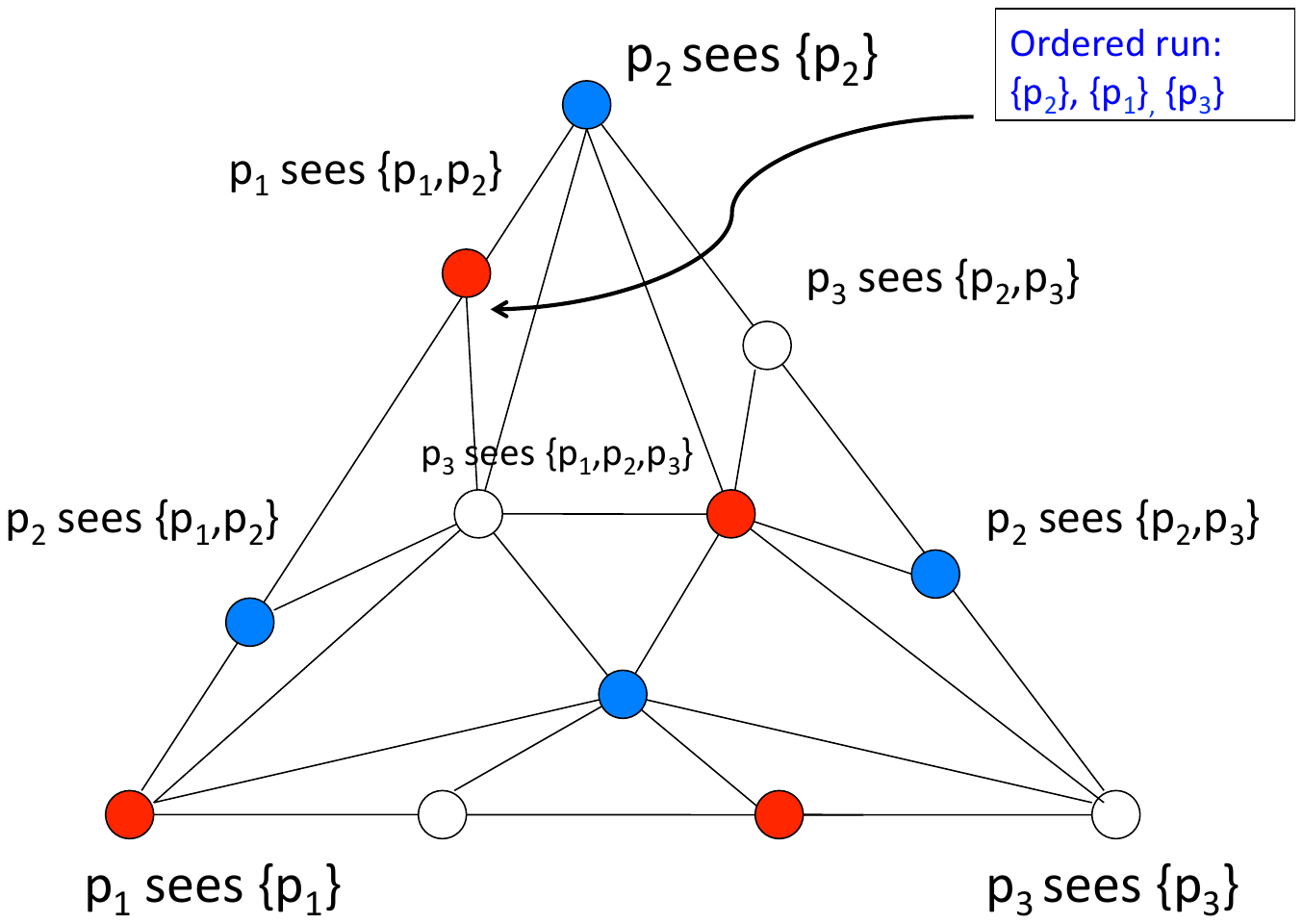}
\caption{\small $\Chr\s$ in the $2$-dimensional case.}
\label{fig:scs}
\end{wrapfigure}

\myparagraph{Iterated tasks for $k$-set consensus.} 
We show that the model of read-write memory and $k$-set consensus
objects can be captured by an iterated \emph{affine} task~\cite{GKM14}. 
The task is defined via a simplicial complex $\RK$, a specific subset of simplices of the second
chromatic subdivision $\Chr^2\s$~\cite{HKR14} in which, intuitively, at
most $k$ processes concurrently \emph{contend}.  (Examples of $\R_1$ and
$\R_2$ for the $3$-process system are given in Figure~\ref{fig:complexes}.) 
We show that the set of IIS runs corresponding to iterations of
this subcomplex $\RK$, denoted $\RK^*$ allows for solving precisely
the same set of tasks as the model of $k$-set consensus does.

Interestingly, our definition of what it means to solve a task in $\RK^*$
requires \emph{every} process to output.
This contrasts with the conventional definition of
task solvability (e.g., using $k$-set-consensus), 
where a failure may prevent a process from producing
an output and, thus, only \emph{correct} processes can be required to output.    
Indeed, $\RK^*$ does not account for process failures: every
process takes infinitely many steps in every run, but, because of the
use of iterated memory,  a ``slow'' process may not be seen by
``faster'' ones from some point.   
The requirement that every process outputs is indispensable in
an iterated characterization of  generic (colored) task solvability that may not allow
one process to ``adopt'' an output of another process.
Indeed, even if ``fast'' processes output, the ``slow'' ones should
be  able to continue in order to enable every correct process to
output in the corresponding $k$-set-consensus model.        
For example, task solvability with consensus objects is captured by 
 the ``total order'' subcomplex $\R_1$ (depicted for $3$ processes in the left part of
 Figure~\ref{fig:complexes}) in which, intuitively,   
\emph{every}  subset of processes should be able to solve consensus.  

Our result is established by the existence of two algorithms. 
The first algorithm solves the \emph{simplex agreement task}~\cite{HS99} on
$\RK$ in the model $k$-concurrency. 
By iterating this
solution, we can implement $\RK^*$ and, thus,  solve any task solvable
in $\RK^*$.
Then, by showing that the $k$-set consensus model solves every task that the
$k$-concurrency model solves, we derive that the former model also
solves every task solvable in $\RK^*$.    
The second algorithm \emph{simulates} runs of a given algorithm using
read-write memory and $k$-set-consensus objects  in
$\RK^*$. 
The simulation is quite interesting in its own right. 
Compared to simulations
in~\cite{HR10,GR10-opodis,GK11,BGK14,GKM14-podc}, our algorithm
ensures that every process eventually outputs in $\RK^*$,
assuming that the simulated algorithm ensures that every correct
process eventually outputs.

Thus, a task is solvable using iterations of $\RK$ if and only if it can
be solved  solvable wait-free using
$k$-set-consensus objects or, equivalently, $k$-concurrently.
Therefore, the $k$-set-consensus model has a bounded 
representation as an \emph{iterated affine task}: processes simply sequentially invoke instances of
$\RK$ for a bounded number of times, until they assemble
enough knowledge to produce an output for the task they are solving.     

Our results suggest a separation between ``natural'' models that have a matching affine task
and, thus, can be captured precisely by a subset of IIS runs and less
``natural'' ones, like WSB, having a manifold structure that is
not affine~\cite{GRH06}.  
%
%
%
We conjecture that such a combinatorial representation can also be
found for a large class of restrictions of the wait-freedom, beyond
$k$-concurrency and $k$-set consensus. 
The claim is supported by a recent derivation of the $t$-resilience
affine task~\cite{SHG16}.


\myparagraph{Related work.}
%
%
There have been several attempts to extend the topological
characterization of~\cite{HS99} to
models beyond wait-free~\cite{HR10,GK11,GKM14-podc}. However, these results
either only concern the special case of \emph{colorless} tasks~\cite{HR10},
consider weaker forms of solvability~\cite{GK11}, or also introduce a new
kind of \emph{infinite} subdivisions~\cite{GKM14-podc}. 

In particular, Gafni et al.~\cite{GKM14-podc} characterized task solvability in
models represented as subsets of IIS runs via \emph{infinite} subdivisions of input
complexes.  This result assumes a limited notion of task solvability
in the iterated model that only guarantees outputs to ``fast''
processes~\cite{Gaf98,RS12,BGK14} that are ``seen'' by every other
process infinitely often.   

In contrast with the earlier work, this paper studies the
inherent combinatorial properties of general
(colored) tasks and assumes the conventional notion of task solvability.  
In a sense, our results for the first time truly capture the combinatorial structure of
a model of computation beyond the wait-free one.


\myparagraph{Roadmap.} 
The rest of the paper is organized as follows. 
Section~\ref{sec:model} gives model definitions, briefly overviews the topological representation of iterated
shared-memory models. 
In Section~\ref{sec:def}, we present the definition of $\RK$
corresponding to the $k$-concurrency model.
In Section~\ref{sec:Rk}, we show that $\RK$ can be implemented
in the $k$-set-consensus model and that any task solvable in 
the $k$-set-consensus model can be solved by iterating $\RK$.
Section~\ref{sec:disc} discusses related models and open questions.

\section{Preliminaries}
\label{sec:model}



%
Let $\Pi$ be a system composed of $n$ asynchronous processes, $p_1,\ldots,p_n$.
We consider two models of communication: (1)~\emph{atomic
snapshots}~\cite{AADGMS93} equipped with \emph{$k$-set consensus objects} and
(2)~\emph{iterated immediate snapshots}~\cite{BG97,HS99}.

\myparagraph{Atomic snapshots and $k$-set consensus.}
The atomic-snapshot (AS) memory is represented as a vector of shared
variables, where processes are associated to distinct vector 
positions, and exports two operations: \emph{update} and \emph{snapshot}. 
An \emph{update} operation performed by $p_i$
replaces the shared variable at position $i$ with a new value and a
\emph{snapshot} returns the current state of the vector.  

The model in which only AS can be invoked is called the AS model.
The model in which can be invoked, in addition to AS, also $k$-set consensus objects,
for some fixed $k\in\{1,\ldots,n-1\}$, 
is called the $k$-set consensus model. 

\myparagraph{Iterated immediate snapshots.}
In the iterated immediate snapshot (IIS) model~\cite{BG97}, 
processes goes through the ordered sequence of
independent memories $M_1$, $M_2$, $\ldots$. 
Each memory $M_r$ is accessed by a process with a single \emph{immediate
snapshot} operation~\cite{BG93a}: the operation performed by
$p_i$ takes a value $v_i$ and
returns a set $V_{ir}$ of values submitted by other processes (w.l.o.g, we assume that values
submitted by different processes are distinct), so that 
the following properties are satisfied: (self-inclusion) $v_i \in V_{ir}$; (containment) $V_{ir}\subseteq
V_{jr}$; and (immediacy) $v_i \in V_{jr}$ $\Rightarrow$ $V_{ir}\subseteq
V_{jr}$. 

 In the IIS communication model, we assume that processes run the
\emph{full-information} protocol: 
the first value each process writes is its \emph{input value}.
For each $r>1$,  the outcome of the
immediate snapshot operation on memory $M_{r-1}$ is submitted 
as the value for the immediate snapshot operation on memory $M_r$.
After a certain number of such (asynchronous) rounds, 
a process may gather enough information to \emph{decide}, i.e.,
to produce an irrevocable non-$\bot$ output value. 
A \emph{run} of the IIS communication model is thus a sequence $V_{ir}$,
$i\in\Nat_n$ and $r\in\Nat$,
determining the outcome of the immediate-snapshot operation for every
process $i$ and each iterated memory $M_r$.

\myparagraph{Failures and participation.}
In the AS or $k$-set consensus model, a process that takes only finitely many steps of the
full-information protocol in a given run is called \emph{faulty}, otherwise it is called
\emph{correct}. A process is called \emph{participating} if it took at
least one step in the computation. We assume that in its first step, a
process writes its input in the shared memory. The set of
participating processes in a given run is called the
\emph{participating set}. Note that, since every process writes its input value in its first step, the inputs of
participating processes are eventually known to every process that
takes sufficiently many steps.

In contrast, the IIS model does not have the notion of a faulty
process. Instead, a process may appear ``slow'' ~\cite{RS12,Gaf98-iis},
i.e., be late in accessing iterated memories from some point on so
that some ``faster'' processes do not see them. 

\myparagraph{Tasks.}
In this paper, we focus on 
distributed \emph{tasks}~\cite{HS99}.
A process invokes a task with an input value and the task returns an output value, so that the inputs and the
outputs across the processes which invoked the task, respect the task
specification.
Formally, a \emph{task} is defined through a set $\I$ of input vectors (one input value for each process),
a set $\O$ of output vectors (one output value for each process), and a total relation $\Delta:\I\mapsto 2^{\O}$
that associates each input vector with a set of possible output vectors. An input $\bot$ denotes a \emph{not
participating} process and an output value $\bot$ denotes an
\emph{undecided} process. Check~\cite{HKR14} for more details on the definition.

\myparagraph{Protocols and runs.}
A \emph{protocol} is a distributed automaton that, for each local
state of a process, stipulates which shared-memory operation and which
state transition the process is allowed to perform in its next step.
We assume here \emph{deterministic} protocols, where only one
operation and state transition is allowed in each state.
A \emph{run} of a protocol is defined as a
sequence of states and shared-memory operations.      

A process is called \emph{active} at the end of a finite run $R$ if
it participates in $R$ but did not returned at the end of $R$.
Let $\textit{active}(R)$ denote the set of all processes that are
active at the end of $R$.     

A run $R$ is \emph{$k$-concurrent} ($k=1,\ldots,n$) if 
at most $k$ processes are \emph{concurrently active} in $R$, i.e.,
$\max\{|\textit{active}(R')|;\; R' \mbox{ prefix of }$R$\}\leq k$.
The \emph{$k$-concurrency} model is the set of $k$-concurrent AS runs.  

\myparagraph{Solving a task.}
A protocol solves a task $T=(\I,\O,\Delta)$ 
 \emph{in the $k$-set-consensus model} (resp., \emph{$k$-concurrently})
if it ensures that in every run of
the $k$-set-consensus model (resp., every $k$-concurrent AS run) in which
processes start with an input vector $I\in\I$, 
(1)~all decided values form a vector $O\in\O$ such that $(I,O)\in\Delta$,
and 
(2)~every correct process decides. 

It is known that the $k$-concurrency model is \emph{equivalent}  to
the $k$-set-consensus model~\cite{GG09}\footnote{In
  fact, this paper contains a self-contained proof of this
  equivalence result.}:
any  task that can be solved $k$-concurrently can also be solved 
in the $k$-set-consensus model, and vice versa.

    
\myparagraph{Standard chromatic subdivision and IIS.}
To give a combinatorial representation of the IIS model, we use the
language of \emph{simplicial complexes}~\cite{Spanier,HKR14}. 
In short, a simplicial complex is defined as a set of \emph{vertices}
and an inclusion-closed set of vertex subsets, called \emph{simplices}. The
dimension of a simplex $\sigma$ is is the number of vertices in it minus one. 
Any subset of these vertices is called a \emph{face} of the simplex.  
A simplicial complex is \emph{pure} (of dimension $n$) if each its
simplices are contained in a simplex of dimension $n$.

A simplicial complex is \emph{chromatic} if it is equipped with a
\emph{coloring} non-collapsing map $\chi$ from its vertices to 
the {\em standard $(n-1)$-simplex} $\s$ of $n$ vertices, 
in one-to-one correspondence with $n$ {\em colors} $1,2, \dots, n$.      
All simplicial complexes we consider here are pure and chromatic. 

Refer to Appendix~\ref{app:topprimer} for more details on the formalism.

For a chromatic complex $C$, we let $\Chr C$ be the subdivision of $C$ obtained by replacing each
simplex in $C$ with its \emph{chromatic subdivision}~\cite{HS99,Koz12, Lin10}. The vertices of $\Chr C$ are pairs $(v, \sigma)$, where $p$
is a vertex of $C$ and $\sigma$ is a simplex of $C$ containing $v$.
vertices $(v_1,\sigma_1)$, $\ldots$, $(v_m,\sigma_m)$ form a simplex
if all $v_i$ are distinct and all $\sigma_i$ satisfy the properties of
immediate snapshots.    
Subdivision $\Chr^1\s$ for the $2$-dimensional simplex $\s$ is given
in Figure~\ref{fig:scs}. Each vertex represents a local state of one
of the three processes $p_1$, $p_2$ and $p_3$ (red for $p_1$, blue for
$p_2$ and white for $p_3$) after it takes a single immediate snapshot.
Each triangle ($2$-simplex) represents
a possible state of the system. A corner vertex corresponds to a local state in
which the corresponding process only sees itself (it took its
snapshot before the other two processes moved). An interior vertex
corresponds to a local state in which the process sees all three
processes. The vertices on the $1$-dimensional faces capture the
snapshots of size $2$.  

If we \emph{iterate} this subdivision $m$ times, each time applying
the same subdivision to each of the simplices, we
obtain the $m^{th}$ chromatic subdivision, $\Chr^m C$.
It turns out that  $\Chr^m \s$ precisely captures the $m$-round
(full-information) IIS model, denoted IS$^m$~\cite{HS99}. 
Each run of IS$^m$
corresponds to a simplex in $\Chr^m \s$.
Every vertex $v$ of $\Chr^m \s$ is thus defined as $(p,
\IS^1(p,\sigma),\ldots,\IS^m(p,\sigma))$, where each $\IS^i(p,\sigma)$
is interpreted as the set of processes appearing in the $i^{th}$ IS
iteration obtained by $p$ in the corresponding $\IS^m$ run.   
The \emph{carrier} of vertex $v$ is then defined as the set of all
processes seen by $p$ in this run, possibly through the views of other processes:
it is the smallest face of $\s$ that contains $v$ in its geometric
realization~\cite{HKR14} (Appendix~\ref{app:topprimer}). 

\myparagraph{Simplex agreement.}
As we show in this paper, the model of $k$-concurrency can be captured
by an iterated \emph{simplex agreement} task~\cite{BG97,HS99}.

Let $L$ be a subcomplex of $\Chr^2\s$.
In the simplex agreement task,  every process starts with the vertex of
$\s$ of its color as an input and finishes with a vertex of $\Chr^m\s$
as an output, so that all
outputs constitute a simplex of $\Chr^2\s$ contained in the face of $\s$ constituted
by the participating processes.     

Formally, the task is defined as $(\s, L,\Delta)$, where, for every face $\t \subseteq \s$, $\Delta(\t) = L \cap \Chr^2
\t$. 
By running~$m$ iterations of this task, we obtain $L^m$,
a subcomplex of $\Chr^{2m}\s$, corresponding to a subset of 
$\IS^{2m}$ runs (each iteration includes two IS rounds).    

%
%

\section{The complex of $k$-set consensus}
\label{sec:def}


We now define $\RK$, a subcomplex of $\Chr^2\s$, that precisely captures the ability of $k$-set consensus (and read-write memory) to solve tasks.
The definition of $\RK$ is expressed via a restriction on the simplices of $\Chr^2\s$
that bounds the size of \emph{contention sets}.
%
Informally, a contention set of a simplex $\sigma\in\Chr^2\s$ 
(or, equivalently, of an $\IS^2$ run) is a set of
processes that ``see each other''. 
When a process $p_i$ starts its $\IS^2$ execution after another process $p_j$
terminates, $p_i$ must observe $p_j$'s input, but not vice versa.
Thus, a set of processes that see each others' inputs must have been
concurrently active at some point.  
Note that processes can be active at the
same time but the immediate snapshots outputs might not permit to
detect it.

Topologically speaking, a contention set of a simplex
$\sigma\in\Chr^2\s$ is a set of processes in $\sigma$ sharing
the same carrier, i.e., a minimal face $\t\subseteq\s$ that
contains their vertices.
Thus, for a given simplex $\sigma\in\Chr^2\s$, the set of contention
sets is defined as follows: 
\begin{definition}[Contention sets]  
\[
\Cont(\sigma) = \{S\subseteq\Pi, \forall p,p' \in S,\carrier(p,\sigma) = \carrier(p',\sigma)\}{}.
 \]
\end{definition}
\begin{figure}[t]
\center
\includegraphics[scale=.85]{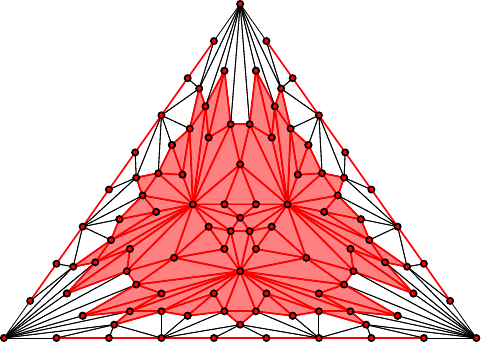}
\caption{\small Contention sets (simplices in red) in a $3$-process system.}
\label{fig:chr2Contention}
\end{figure}

Contention sets for simplices of $\Chr^2\s$ in a $3$-process system are
depicted in Figure~\ref{fig:chr2Contention}: for each simplex $\sigma\in\Chr^2\s$,
every face of $\sigma$ that constitutes a red simplex is a contention set of
$\sigma$.
In an interior simplex, every set of processes are contention sets.
Every ``total order'' simplex (shown in blue in
Figure~\ref{fig:R1}), matching a run in which processes proceed, one by one,
in the same order in both $\IS^1$ and~$\IS^2$, has only three singleton as contending sets. 
All other simplices include a contention set of two processes which consists of the vertices at the boundary.



Now $\RK$ is defined as the set of all 
simplices in $\Chr^2\s$, in which 
the contention sets of have cardinalities at most $k$:
\begin{definition}[Complex $\RK$]
\[
\RK = \{\sigma \in \Chr^2\s,\forall S\in \Cont(\sigma),|S|\leq k\} {}.
\]
\end{definition}
It is immediate that the set of simplices in $\RK$ constitutes a
simplicial complex: every face $\tau$ of $\sigma\in\RK$ is also in $\RK$.  

\begin{figure}[t]
\captionsetup[subfigure]{justification=centering}
  \begin{minipage}[b]{.49\linewidth}
    \centering
		\includegraphics[scale=1.1]{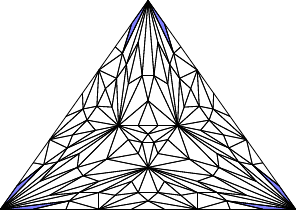}
    \subcaption{Complex $\R_1$}\label{fig:R1}
  \end{minipage}
  \hfill
  \begin{minipage}[b]{.49\linewidth}
    \begin{center}
    		\includegraphics[scale=1.2]{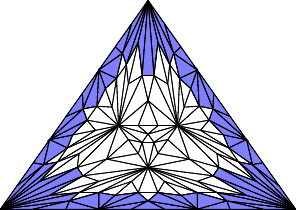}
    \subcaption{Complex $\R_2$}\label{fig:R2}
    \end{center}
  \end{minipage}
  \caption{$\R_1$ and $\R_2$ (in blue) for $3$ processes.}
  \label{fig:complexes}
\end{figure}

Examples of $\R_1$ and $\R_2$ for a $3$-process system are shown in
Figures~\ref{fig:R1} and~\ref{fig:R2}, respectively.
Obviously, for the unrestricted $3$-set consensus case, $\R_3=\Chr^2\s$. 
Note that $\R_1$ only contains six ``total order'' simplices, while
$\R_2$ consists of all simplices of $\Chr^2\s$ that touch the boundary.


%


%

\section{From $k$-set consensus to $\RK^*$ and back}
\label{sec:Rk}

We show that any task solvable with $k$-set consensus (and read-write
shared memory) can be solved in $\RK^*$, and vice versa.
The main result is then established via \emph{simulations}: a run of
an algorithm solving a task in one model is simulated in the other.

\subsection{From $k$-set consensus to $k$-concurrency}
\label{subsec:Kset<->Kconc}

We first show that a $k$-concurrent shared memory system is
equivalent, regarding task solvability, to a shared memory system
enhanced with $k$-set consensus objects.
The result has been stated in a technical report~\cite{GG09}, but no explicit proof has
been given available in the literature until now, and we fill the gap below.
For the sake of completeness and to make referencing simpler, we
propose here a direct simulation with proofs.

\myparagraph{Simulating a $k$-process shared memory system.}
%
We employ \emph{generalized state machines} (proposed in~\cite{GG11-univ} and
extended in~\cite{RST14}) that allow for simulating a $k$-process read-write
memory system in the $k$-set-consensus model.      
To ensure consistency of simulated read and write operations, we use 
\emph{commit-adopt} objects~\cite{Gaf98} that can be implemented using
reads and writes.
A commit-adopt object exports one operation 
 $\textit{propose}(v)$ that takes a parameter in an arbitrary range
and returns a couple
$(\emph{flag},v')$, where $\textit{flag}$ can be either
$\textit{commit}$ or $\textit{adopt}$ and
where $v'$ is a previously proposed value.
Moreover, if a process returns a $\textit{commit}$
flag, then every process must return the same 
value.
Further, if no two processes propose different values, then
all returned flags must be $\textit{commit}$.

Liveness of the simulation relies on calls to \emph{$k$-simultaneous
  consensus} objects~\cite{AGRRT10}.
To access a $k$-simultaneous consensus object, a process proposes a
vector of $k$ inputs, one for each of the  \emph{consensus instances},
$1,2,\ldots,k$,
and the object returns a couple $(i,v)$, where \emph{index} $i$
belongs to $\{1,\ldots,k\}$
and $v$ is a value proposed by some process at index $i$. 
It ensures that no two processes obtain different values with the same index.
Moreover, if $\ell\leq k$ distinct input vectors are proposed then
only values at indices $1,\ldots,\ell$ can be output. 
The $k$-simultaneous consensus object is equivalent to $k$-set-consensus in a
read-write shared-memory system~\cite{AGRRT10}. 

\begin{algorithm}[t]
 \caption{$k$ processes shared memory system simulation: process $p_i$\label{Alg:K-conc}}
 \begin{small}
\SetKwRepeat{Repeat}{Repeat forever}{End repeat}%
$\mathbf{Shared Objects}$: $KSC[1\dots]$ $\mathbf{:\ k-simultaneous\ consensus\ objects}$\;
$CA[1\dots][1\dots k]$ $\mathbf{:\ commit-adopt\ objects}$\;
$MEM[1\dots n][1\dots k]$ $\mathbf{init}$ $(-1,\bot)$ $\mathbf{:\ single\ writer\ shared\ memory\ array}$\;
$\mathbf{Init}$: $r_i \leftarrow 0;$ \lForEach{$m \in \{1,\dots,k\}$}{$(WC_i[m],\textit{View}_i[m]) \leftarrow (0,\emptyset)$}

\vspace{1em}

\Repeat{}{
	$r_i \leftarrow r_i+1$\;
	$(\textit{Index}_i,\textit{Value}_i) \leftarrow KSC[r_i].\textit{propose}(WC,\textit{View}_i)$\;\label{Alg:Kc:SC}
	$(Flag_i[\textit{Index}_i],val_i) \leftarrow CA[r_i][\textit{Index}_i].\textit{propose}(\textit{Value}_i)$\;\label{Alg:Kc:1CA}
	\lIf{$val_i=(c,*) \mathbf{\ with\ } c\geq WC_i[\textit{Index}_i]$}{$(WC_i[\textit{Index}_i],\textit{View}_i[\textit{Index}_i]) \leftarrow val_i$}
	\ForEach{$m \in \{1,\dots,k\} \setminus \textit{Index}_i$}{
		$(Flag_i[m],val_i) \leftarrow CA[r_i][m].\textit{propose}(\textit{View}_i[m])$\;\label{Alg:Kc:OtherCA}
		\lIf{$val_i=(c,*) \mathbf{\ with\ } c\geq WC_i[m]$}{$(WC_i[m],\textit{View}_i[m]) \leftarrow val_i$}
	}	
	\ForEach{$m \in \{1,\dots,k\}$}{
		\If{$Flag_i[m]=Commit$}{\label{Alg:Kc:Commit}
			$MEM[i][m].\textit{Update}(WC_i[m],\textit{WriteVal}(WC_i[m],\textit{View}_i[m]))$\;\label{Alg:Kc:Write}
			$WC_i[m] \leftarrow WC_i[m]+1$\;\label{Alg:Kc:WCincr}
			$\textit{View}_i[m]=CurWrites(MEM.\textit{Snapshot}())$\;\label{Alg:Kc:Mainselect}
		}
	}	
}

\vspace{1em}

$\mathbf{With}\ CurWrites(MEM_{val}) =$\label{Alg:Kc:CurWrites1}\\
\lForEach{$m \in \{1,\dots,k\}$} {$curWC[m]=-1$, $curWrite[m]=\bot$}
\ForEach{$(m,l) \in \{1,\dots,k\}\times\{1,\dots,n\}$} {
	\If{$MEM_{val}[l][m].WC>curWC[m]$}{
	$curWC[m]=MEM_{val}[l][m].WC$, $curWrite[m]=MEM_{val}[l][m].Value$\;
	}
}
$\mathbf{return}$ $curWrite$\;\label{Alg:Kc:CurWrites2}
\end{small}
\end{algorithm}

Our simulation is described in  Algorithm~\ref{Alg:K-conc}.
We use three shared abstractions: an infinite array of
$k$-simultaneous consensus objects $KSC$, an infinite array of arrays
of $k$ indexed commit-adopt objects $CA$, and a single-writer
multi-reader memory $MEM$
with $k$ slots.
%

In every round, processes use the corresponding $k$-simultaneous
consensus object first (line~\ref{Alg:Kc:SC}) and then go through the set of $k$
commit-adopt objects (lines~\ref{Alg:Kc:1CA}--\ref{Alg:Kc:OtherCA}),
starting with the index output by the $k$-simultaneous consensus object
(line~\ref{Alg:Kc:1CA}). It is guaranteed
that at least one process commits, in particular, 
process $p_j$ that is the first to return from its first commit-adopt
invocation in this round (on a commit-object $C$), because
any other process  with a
different proposal must access a different commit-adopt object first
and,  thus, must invoke $C$ after $p_j$ returns. 
To ensure that a unique written value is selected, processes replace
their current proposal values with the value adopted by the
commit-adopt objects (lines~\ref{Alg:Kc:1CA}--\ref{Alg:Kc:OtherCA}).
Note that the processes do
not select values corresponding to an older  round of simulation, to
ensure that processes
do not alternate committing and adopting the same value indefinitely.

In the simulation, the \emph{simulating} processes propose snapshot results for the
\emph{simulated} processes. Once a proposed snapshot has been committed,
a process  stores in the shared memory the 
value that the simulated process must write in its next step (based on
its simulated algorithm), equipped with the corresponding \emph{write
counter}
(line~\ref{Alg:Kc:Write}). The write counter is then incremented and a
new snapshot proposal is computed (line~\ref{Alg:Kc:Mainselect}).
To compute a simulated snapshot,
for each process, we select the most recent value available in the
memory $MEM$ by comparing the write counters $WC$ (auxiliary function
$\textit{CurWrites}$ at
lines~\ref{Alg:Kc:CurWrites1}--\ref{Alg:Kc:CurWrites2}).

\begin{lemma}
Algorithm~\ref{Alg:K-conc} provides a non-blocking simulation of a
$k$-process read-write shared-memory system in the $k$-set consensus
model.
Moreover, if there are $\ell<k$ active processes, then 
one of the first $\ell$ simulated processes is guaranteed to make progress.\label{lem:KprocSim} 
\end{lemma}

The proof of Lemma~\ref{lem:KprocSim} can be found in
Appendix~\ref{app:Proof_KprocSim}.
The proof is constructed by showing that:
(1) No two different written values are computed for the same
simulated process and the same write counter;
(2) At every round of the simulation, at least one simulator commits
a new simulated operation; (3) Every committed simulated snapshot
operation can be linearized at the moment when the actual snapshot operation which
served for its computation took place; and (4) Every
simulated write operation can be linearized to the linearization time
of the first actual write performed by a simulator with the
corresponding  value.

\myparagraph{Using the extended BG-simulation to simulate a $k$-concurrent execution.}
We have shown that a $k$-process read-write shared memory system can
be simulated in the $k$-set-consensus model. Now we show that this simulated system
can be used to simulate $k$-concurrency.
The idea is to make the obtained $k$-process system run a \emph{BG-simulation}
protocol~\cite{BG93b,BGLR01}, so that  at most $k$ simulated processes are active at a time.

The BG-simulation technique 
allows $k+1$ processes $s_1,\ldots,s_{k+1}$, called
\emph{BG-simulators}, to wait-free simulate a \emph{$k$-resilient}
execution of any protocol $\A$ on $m$ processes $p_1,\ldots,p_m$
($m>k$). The simulation guarantees that each simulated step of every
process $p_j$ is either agreed upon by all simulators, or one less
simulator participates further in the simulation for each step which
is not agreed on (in this, we say that the step simulation is
\emph{blocked} because of the faulty or slow simulator). 

%
The technique was later turned into
\emph{extended BG-simulation}~\cite{Gaf09-EBG}. 
The core of this technique is
the \emph{Extended Agreement} (EA) algorithm, which ensures safety of
consensus but not necessarily liveness:
it may \emph{block} if some process has
slowed down in the middle of its execution. 
Additionally, the EA protocol exports an \emph{abort} operation that, when applied to a
blocked EA instance, re-initializes it so that it can move forward
until an output is computed or another process makes it block again.

Our simulation is quite simple. Before running the $k$-process
simulation using Algorithm~\ref{Alg:K-conc}, processes write their input
states in the memory. The $k$-process simulation is used to run an \emph{extended BG-simulation} that
executes the code of the task solution for the simulated initial
processes. Once  a task output for a simulated process is available in
the shared memory, the process stops participating in the $k$-process
simulation.
A process that completed its initial write but has not yet been
provided with a task output is called \emph{active},
i.e., this process is both available to be simulated by the
BG-simulators and is participating in the simulation of the $k$ BG-simulators. 

In our case, we use the \emph{extended BG-simulation} in a slightly
different manner than in the original paper~\cite{Gaf09-EBG}.
Instead of running processes in
lock-step as much as possible, by selecting the least advanced
available process (\emph{breadth-first} selection), the BG-simulators run the
processes with as low concurrency as possible by selecting the most
advanced available process (\emph{depth-first} selection). To prevent
simulators from getting blocked on all active processes, a
BG-simulator stops participating if the number of \emph{active}
processes is strictly lower than its \emph{identifier}
(the index of the simulated process the BG-simulator is executed on, from $1$ to $k$).
If a \emph{BG-simulator} is blocked on all \emph{active}
simulated processes, but has an identifier lower or equal to their
number,  it uses the \emph{abort} mechanism to exclude
\emph{BG-simulators} with large identifiers that should be stopped.

\begin{lemma}
All tasks solvable in the $k$-concurrency model can be solved in the $k$-set-consensus model.\label{lem:kSet->kConc}
\end{lemma}
\begin{proof}
The BG-simulators select the most advanced \emph{available} (not blocked by
a BG-simulator) process to execute, thus, a
process never simulated yet is selected if and only if all currently 
started simulations of \emph{active} processes are \emph{blocked}. But
at most $m$ simulated codes can be blocked by $m$ BG-simulators
($m\leq k$), thus,
at most $k$ \emph{active} processes can be concurrently
simulated and at least of them is not blocked.
Moreover, when there are $\ell<k$ \emph{active} processes,
then progress is guaranteed to one of the first $\ell$ BG-simulators,
see Lemma~\ref{lem:KprocSim}. The remaining $k-\ell$ BG-simulators stop
participating in the simulation and cannot block the $\ell$
first ones, as they are eventually excluded using the \emph{abort}
mechanism. The \emph{abort} mechanism is used only finitely many
times, only once a BG-simulator witnesses that the number of
\emph{active} processes has decreased (and there are finitely many
processes). Therefore, as long as there are \emph{active correct}
processes, the \emph{BG-simulation} \emph{makes progress} and
eventually every \emph{correct} process obtains an output in the task solution.
\end{proof}

\subsection{From $k$-concurrency to $\RK$.}
\label{subsec:Kset->RK} 
We now show that $k$-concurrency can
\emph{solve} $\RK$, i.e., it can solve the chromatic simplex agreement task
on the subcomplex $\RK$.
\begin{lemma}
A $k$-concurrent execution of two rounds of any immediate snapshot
algorithm solves the simplex agreement task on $\RK$.
\label{lem:SimplexAgreementRK}
\end{lemma}
\begin{proof}
Let us consider a set of $\IStwo$ outputs provided by a $k$-concurrent
execution of any IS algorithm (e.g., \cite{BG93a}): 
the set of $\IStwo$ outputs forms a valid simplex $\sigma$ in
$\Chr^2\s$~\cite{BG97}, as the set of $k$-concurrent runs is a subset of the
wait-free runs. Let us consider a contention set $S$ containing
two processes $p$ and $q$, and let us assume that $p$ and $q$ were
never executed concurrently during their executions of the two rounds
of immediate snapshots. Without loss of generality, we can consider
that $p$'s computation was terminated before the activation of $q$, so
$p$ cannot be aware of $q$'s input as it did not perform any operation
before $p$ finishes the two rounds of immediate snapshots. Thus, $p$
cannot see $q$, which contradicts the contention set
definition. Therefore all processes in a contention set were
\emph{active} at the same time during the execution, hence a
$k$-concurrent execution implies that contention sets cannot contain more than $k$ processes.
\end{proof}

It is easy to complete this result by showing that the $k$-set consensus model is, regarding task solvability, at least as strong as the $\RK^*$ model:

\begin{theorem}\label{thm:RK->kset}
Any task solvable by $\RK^*$ can be solved in the $k$-set consensus model.
\end{theorem}
\begin{proof}
As shown in Lemma~\ref{lem:SimplexAgreementRK}, the simplex agreement
task on $\RK$ is solvable in the $k$-concurrency model. Moreover,
according to
Lemma~\ref{lem:kSet->kConc}, any task solvable in the $k$-concurrency
model can be solved in the $k$-set consensus model, and hence in particular, the
simplex agreement task on $\RK$. Therefore, by iterating a solution to
the simplex agreement task on $\RK$, a run of $\RK^*$ can be simulated
in the $k$-set consensus model and
used to solve any task solvable in $\RK^*$.
\end{proof}

\subsection{From $\RK^*$ to $k$-set consensus}
\label{subsec:RK->Kset}
Now we show how to simulate in $\RK^*$ any algorithm that uses read-write memory
and $k$-set-consensus objects.

\myparagraph{$k$-set consensus simulation design.}
Making a \emph{non-blocking} simulation of read-write memory can be
trivially done in $\RK^*$, since the set of $\RK^*$ runs is a subset
of $\IS^*$ runs, and there exists several algorithms simulating read-write memory
in $\IS^*$, e.g.,~\cite{GR10-opodis}. 

Solving $k$-set agreement is also not very complicated: every
iteration of $\RK$ provides a set of at most $k$ \emph{leaders}, i.e.,
processes with an \ISone output containing at most $k$ elements, where
at least one such \emph{leader} is visible to every process, i.e., it
can be identified as a \emph{leader} and its input is visible to
all. The set of leaders of $\mathcal{R}_2$ are shown in 
figure~\ref{fig:leader} in red, it is easy to observe that every simplex 
in $\mathcal{R}_2$ has at most two leaders, and that one is visible to 
every process (every process with a carrier of size at most 2 is a 
leader). This property gives a very simple $k$-set agreement algorithm:
every process decides on the value proposed by one of these $k$
leaders. We will later show how this property can be derived
from the restriction of $\RK$ on the size of \emph{contention sets}.

The difficulty of the simulation consists mostly in combining the
shared-memory and $k$-set agreement simulation, as some processes may
be accessing distinct agreement objects while other processes are
performing read-write operations. Indeed, liveness of our
$k$-set-agreement algorithm relies on the \emph{participation} of
visible leaders, i.e., on the fact that the leaders propose values for
this instance of $k$-set agreement. In this sense,  our $k$-set
agreement algorithm may block if some process is performing a
read-write operation or is involved in a different instance of $k$-set
agreement. Likewise, the read-write memory simulation is only
\emph{non-blocking}, so it can be indefinitely blocked by a process
waiting to complete an agreement  operation.

The solution we propose consists in (1)~synchronizing the two simulations,
in order to ensure that, eventually, at least one process will
complete its pending operation, and (2)~ensuring that the processes
collaborate by participating in every simulated operation. In our
solution, every process tries to propagate every observed proposed
value (for a write operation), and every process tries to reach an
agreement in every $k$-set-agreement object accessed by some process.
For that, we make the processes participate in both simulation protocols
(read-write and $k$-set agreement) in every round, until they decide. 

Even though the simulated algorithm executes only one
operation at a time and requires the output of the previous operation
to compute the input for the following one, we enrich the
simulated process with \emph{dummy} operations that do not alter the
simulation result. Then eventually some \emph{undecided} process is
guaranteed to complete both pending operations, where at most one of
them is a \emph{dummy} one. This scheme provides a \emph{non-blocking}
simulation of any algorithm using read-write shared memory and $k$-set
agreement objects.

\begin{figure}[t]
\captionsetup[subfigure]{justification=centering}
  \begin{minipage}[b]{.49\linewidth}
    \centering
		\includegraphics[scale=1.2]{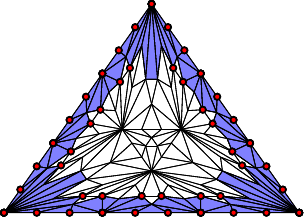}
    \subcaption{Leaders.}\label{fig:leader}
  \end{minipage}
  \hfill
  \begin{minipage}[b]{.49\linewidth}
    \begin{center}
    		\includegraphics[scale=1.2]{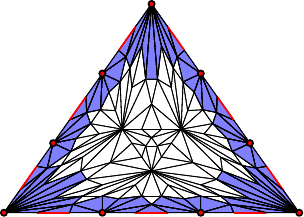}
    \subcaption{Processes with the smallest $\IS^2$ output.}\label{fig:smallestIS2snap}
    \end{center}
  \end{minipage}
  \caption{$\R_2$ for $3$ processes: (a)~leaders --- vertices in red, and (b)~processes with the
    smallest $\IS^2$ output --- simplices in red.}
  \label{fig:LeaderMinSnap}
\end{figure}

The shared memory simulation from~\cite{GR10-opodis} provides progress
to the processes with the smallest snapshot output, while our $k$-set
agreement algorithm provides progress to the 
leader with the smallest $\RK$ output, i.e., the smallest $\IS^2$ output.
We leverage these properties by running the read-write simulation
only on the outputs of $\RK$  (i.e., in every second round of
immediate snapshots). In the $2$-dimensional case, the set of
processes with the smallest $\IS^2$ outputs are presented as red
simplicies in Figure~\ref{fig:smallestIS2snap} for $\mathcal{R}_2$.
This way we guarantee that
at least the leader with the smallest $\RK$ output
will make progress in both simulations.
Indeed, the definition of $\RK$ implies that 
the set of processes with the smallest $\RK$ outputs includes a leader. 
Figure~\ref{fig:LeaderMinSnap} gives an example of
an intersection between the set of processes with the smallest 
$\IS^2$ output and the set of leaders: here every process
with the smallest $\IS^2$ output has a carrier of size at most $2$ and every such process is a leader.

\myparagraph{$k$-set consensus simulation algorithm.}
Algorithm~\ref{Alg:K-cons} provides a simulation of any 
algorithm using read-write shared memory (w.l.o.g., atomic snapshots)
and $k$-set-agreement objects. The algorithm is based on the shared
memory simulation from~\cite{GR10-opodis}, applied on \IStwo outputs
of every iteration of $\RK$,
combined with a parallel execution of instances of our $k$-set
agreement algorithm.
The simulation works in rounds that can be decomposed into three stages:
communicating through $\RK$, updating local information, and validating progress.  

\begin{algorithm}[t]
 \caption{k-set consensus simulation in $\RK^*$: process $i$\label{Alg:K-cons}}
 \begin{small}
\SetKwRepeat{Repeat}{Repeat forever}{End repeat}%
$\mathbf{Init}$: $r_i \leftarrow 0; \textit{State}_i \leftarrow undecided; \textit{ConsId}_i\leftarrow \bot;\textit{ConsProp}_i\leftarrow \bot$\;
$\textit{WriteVal}_i[i] \leftarrow \mathbf{FirstWrite_i()};\textit{WriteCount}_i[i]\leftarrow 1$\;
\lForEach{$m \in \{1,\dots,n\}\setminus \{i\}$}{$(\textit{WriteCount}_i[m],\textit{WriteVal}_i[m]) \leftarrow (0,\bot)$} 
$ConsHistory_i \leftarrow \emptyset
\mathbf{\ :\ List\ of\ adopted\ agreement\ proposals}$\;

\vspace{1em}

\Repeat{}{
	$r_i \leftarrow r_i+1$;$\textit{Leaders}_i\leftarrow true$\;				
	$\IStwo_{output} = \IStwo[r_i](\textit{State}_i,(\textit{WriteCount}_i,\textit{WriteVal}_i),ConsHistory_i)$\;\label{Alg:K-cons:l:RK}

	\vspace{1em}	
	
	\ForEach{$(j,View_j) \in \IStwo_{output}$}{\label{Alg:K-cons:l:UpdateMin}
		$\mathbf{Let\ } (\textit{State}_j,(\textit{WriteCount}_j,\textit{WriteVal}_j),ConsHistory_j) \leftarrow \mathbf{RKInput}(j)$\;\label{Alg:K-cons:l:SelectInput}
		\ForEach{$m \in \{1,\dots,n\}$} {\label{Alg:K-cons:l:UpdateWriteMin}
			\If{$\textit{WriteCount}_j[m]>\textit{WriteCount}_i[m]$}{
				$\textit{WriteCount}_i[m]=\textit{WriteCount}_j[m],\textit{WriteVal}_i[m]=\textit{WriteVal}_j[m]$\;
			}
		}\label{Alg:K-cons:l:UpdateWriteMax}
			\If{$|\mathbf{Undecided}(View_j)|\leq k$} {\label{Alg:K-cons:l:Leaders}
			\lIf{$\nexists (\textit{ConsId}_i,*) \in ConsHistory_j$}{$\textit{Leaders}_i \leftarrow false$}
			\ForEach{$(A_{id},A_{val})\in ConsHistory_j$}{
				$\mathbf{ReplaceOrAdd\ } (A_{id},*) \mathbf{\ in\ } ConsHistory_i \mathbf{\ with\ } (A_{id},A_{val})$\;\label{Alg:K-cons:l:AdoptLeader}
			}
			}
	}	\label{Alg:K-cons:l:UpdateMax}	
	
	\vspace{1em} 
	
	\If{$(\Sigma_{m\in\{1,\dots,n\}} WritesCount_i[m])=r_i$ }{\label{Alg:K-cons:l:testCount}
		\If{$\mathbf{PendingWriteSnapshotOperation()}$} { 
		 	$\mathbf{TerminateWriteOperation}(\textit{WriteVal}_i)$\;\label{Alg:K-cons:l:ValidateWrite}
		}
		\If{$\textit{Leaders}_i\wedge \textit{ConsId}_i\neq\bot$}{\label{Alg:K-cons:l:Decide}
			$\textit{ConsProp}_i \leftarrow A_{val}  \mathbf{\ where\ }  (\textit{ConsId}_i,A_{val})\in ConsHistory_i$\;
			$\mathbf{TerminateAgreementOperation}(A_{val})$; $\textit{ConsId}_i\leftarrow \bot$\;
		}\label{Alg:K-cons:l:DecideMax}
		\lIf{$\mathbf{Terminated}()$}	{$\textit{State}\leftarrow decided$}\label{Alg:K-cons:l:Terminate}
		\Else{\label{Alg:K-cons:l:NewOpMin}
			$\textit{WriteCount}_i[i] \leftarrow \textit{WriteCount}_i[i] +1$\;\label{Alg:K-cons:l:Increment}
			\If{$\mathbf{NextAgreementOperation()}=\mathbf{Available}$} {
				$(\textit{ConsId}_i,\textit{ConsProp}_i)\leftarrow \mathbf{NextAgreement_i()}$\;\label{Alg:K-cons:l:NewAgreement}
				\If{$(\nexists (A_{id},A_{val}) \in ConsHistory_i \mathbf{\ with\ } A_{id} = \textit{ConsId}_i)$}{ \label{Alg:K-cons:l:PreventingAdopt}	
					$\mathbf{Add\ } (\textit{ConsId}_i,\textit{ConsProp}_i) \mathbf{\ in\ } ConsHistory_i$\;\label{Alg:K-cons:l:AdoptFirst}
					}
			}
			\If{$\mathbf{NextWriteSnapshotOperation()}=\mathbf{Available}$}{
				$\textit{WriteVal}_i[i] \leftarrow \mathbf{NextWrite_i()}$\;
			}\label{Alg:K-cons:l:NewOpMax}
		}
	}
}
\end{small}
\end{algorithm}

The first stage consists in accessing the new $\RK$ iteration
associated with the round, using information on the ongoing operations as
an input (see line~\ref{Alg:K-cons:l:RK}). For memory operations, two
objects are contained in $\RK$'s input, an array containing the most
recent known write operations for every process, $\textit{WriteVal}_i$, and a
timestamp associated with each process write value, $\textit{WriteCount}_i$. A
single object is used for the agreement operations, $\textit{ConsHistory}_i$, a
list of all adopted proposals for all accessed agreement
objects. Finally, a value $\textit{State}$, set to $\textit{decided}$ or $\textit{undecided}$, is
put in $\RK$'s input, to indicate whether the process has completed its simulation.

The second stage consists in updating the local information according
to the output obtained from $\RK$ (lines~\ref{Alg:K-cons:l:UpdateMin}--\ref{Alg:K-cons:l:UpdateMax}). The input value of each process
observed in the second immediate snapshot of $\RK$ is extracted
(line~\ref{Alg:K-cons:l:SelectInput}). These selected inputs are
examined in order to replace the local write values $\textit{WriteCount}_i$ with the most recent ones, i.e., associated
with the largest write counters 
(lines~\ref{Alg:K-cons:l:UpdateWriteMin}--\ref{Alg:K-cons:l:UpdateWriteMax}).The
$ConsHistory$ variable of every leader, i.e., a process with an
$\ISone$ output containing at most $k$ \emph{undecided} process inputs
(using the variable $\textit{State}$), is scanned in order to adopt all its
decision estimates
(lines~\ref{Alg:K-cons:l:Leaders}--\ref{Alg:K-cons:l:AdoptLeader}). Moreover,
$\textit{Leaders}_i$ boolean value is used to check if every observed leader
transmitted a decision estimate for the pending agreement operation, $\textit{ConsId}_i$.

The third stage consists in checking whether pending operations can
safely be terminated (lines~\ref{Alg:K-cons:l:testCount}--\ref{Alg:K-cons:l:DecideMax}), and if so, whether the process has
completed its simulation (line~\ref{Alg:K-cons:l:Terminate}) or if new
operations can be initiated (line~\ref{Alg:K-cons:l:NewOpMin}--\ref{Alg:K-cons:l:NewOpMax}). 

Informally, it is safe for a process
to decide in line~\ref{Alg:K-cons:l:Decide}, as there are at most $k$
\emph{Leaders} per round, one of which is (1)~visible to every process
and (2)~provides a decision estimate for the pending agreement.
Thus every process adopts the decision estimate from a \emph{leader}
of the round, reducing the set of possible distinct decisions to
$k$.

A pending memory operation terminates when the round number
$r_i$ equals the sum of the currently observed write counters (test at
line~\ref{Alg:K-cons:l:testCount}), as in the original
algorithm~\cite{GR10-opodis}. Indeed, the equality
implies that the writes in the estimated snapshot have been observed by every
process (line~\ref{Alg:K-cons:l:ValidateWrite}). Last, if a
process did not terminate, it increments its write counter and, if
there is a new operation available, the process selects the operation (see
lines~\ref{Alg:K-cons:l:Increment}--\ref{Alg:K-cons:l:NewOpMax}).

If there is a new agreement operation, then the input proposal and the object
identifier are selected (line~\ref{Alg:K-cons:l:NewAgreement}) and they are
used for the current decision estimate in $\textit{ConsHistory}_i$
(line~\ref{Alg:K-cons:l:AdoptFirst}), unless a value has already been
adopted (line~\ref{Alg:K-cons:l:PreventingAdopt}). If there is a new
write operation then the current write value is simply changed
(line~\ref{Alg:K-cons:l:NewOpMax}), a \emph{dummy} write thus consists 
in re-writing the same value. \footnote{Note that our agreement algorithm is far from efficient
  for multiple reasons. Progress could be validated at every round
  and not only when a write is validated. Moreover, processes could
  also preventively decide the output for objects not yet
  accessed. Lastly, processes could also  adopt proposals from
  \emph{non-leaders} when no visible leader has a proposition.}

\begin{lemma}
In $\RK^*$, Algorithm~\ref{Alg:K-cons} provides a non-blocking simulation of any shared memory algorithm with access to $k$-set-agreement objects.\label{lem:KconsSim} 
\end{lemma}

The proof of Lemma~\ref{lem:KconsSim} is delegated to
Appendix~\ref{app:Proof_KprocSim}. The main aspects of the proof are
taken from the base algorithm from~\cite{GR10-opodis}, while the
liveness of the agreement objects simulation relies on the restriction
provided by $\RK^*$ and the maximal size of \emph{contention
  sets}. 
  
  Lemma~\ref{lem:KconsSim} implies the following result:

\begin{theorem}
Any task solvable in the $k$-set-consensus model can be solved in $\RK^*$\label{thm:RKtoKcons} 
\end{theorem}

\begin{proof}
To solve in $\RK^*$ a task solvable in the $k$-set-consensus model, we can
simply use Algorithm~\ref{Alg:K-cons}, simulating any given algorithm
solving the task in the $k$-set-consensus model.

The non-blocking simulation provided by Algorithm~\ref{Alg:K-cons}
ensures, at each point, that at least one live process eventually terminates.
As there are only finitely many processes, every live process eventually terminates. 
\end{proof}
Lemma~\ref{lem:kSet->kConc}, Theorem~\ref{thm:RK->kset}, and
Theorem~\ref{thm:RKtoKcons} imply the following equivalence result: 

\begin{corollary}
The $k$-concurrency model, the $k$-set-consensus model, and $\RK^*$ are equivalent regarding task solvability.
\label{cor:equiv}
\end{corollary}

\section{Concluding remarks: on minimality of $\Chr^2\s$ for $k$-set consensus}
\label{sec:disc}

\begin{wrapfigure}{r}{0.25\textwidth}
\center
\includegraphics[scale=0.8]{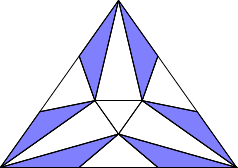}
\caption{\footnotesize Fully ordered sub-$\Chr\s$}
\label{fig:fullOrderCh1}
\end{wrapfigure}
This paper shows that the models of $k$-set consensus and
$k$-concurrency are captured by the same affine task $\RK$,   
%
defined as a subcomplex of $\Chr^2\s$.
One may wonder if there exists a simpler equivalent affine task, defined as a
subcomplex of $\Chr\s$, the $1$-degree of the standard chromatic subdivision. 
%
To see that this is in general not possible, consider the case of
$k=1$ (consensus) in a $3$-process system.
We can immediately see that the corresponding subcomplex of $\Chr\s$
must contain all ``ordered'' simplexes depicted in
Figure~\ref{fig:fullOrderCh1}.
Indeed, we must account for a wait-free $1$-concurrent~$\IS^1$
run in which, say, $p_1$ runs first until it completes (and it must
outputs its corner vertex in $\Chr\s$), then $p_2$ runs alone until it
outputs its vertex in the interior of the  face $(p_1,p_2)$ and, finally, $p_3$
must output its interior vertex.    

The derived complex is connected. 
Moreover, any number of its iterations still results in a connected
complex. The simple connectivity  argument implies that consensus
cannot be solved in this iterated model and, thus, the complex cannot
capture \mbox{$1$-concurrency}.  

Interestingly, the complex in Figure~\ref{fig:fullOrderCh1} precisely
captures the model in which, instead of consensus, weaker
\emph{test-and-set} (TS) objects are used: (1)~using TS, one easily make
sure that at most one process terminates at an $\IS$ level, and
(2)~in $IS$ runs defined by this subcomplex, 
any pair of processes can solve consensus using this complex and, thus, a TS object
can be implemented.   
It is not difficult to generalize this observation to
\emph{$k$-TS} objects~\cite{MRT06}: the corresponding complex consists of
all simplices of $\Chr\s$, contention sets of which are of size at
most $k$. The equivalence (requiring a simple generalization
for the backward direction) can be found in~\cite{MRT06,GRT07}.

Overall, this raises an intriguing question whether every object, when used in the
read-write system, can be captured via a subcomplex of $\Chr^m\s$ for
some $m\in\Nat$.




\def\noopsort#1{} \def\No{\kern-.25em\lower.2ex\hbox{\char'27}}
  \def\no#1{\relax} \def\http#1{{\\{\small\tt
  http://www-litp.ibp.fr:80/{$\sim$}#1}}}


\newpage
\appendix 

\section{Simplicial complexes}
\label{app:topprimer}

We review now several notions from topology. For more detailed coverage of the topic please refer
to~\cite{Spanier,HKR14}.

A {\em simplicial complex} is a set $V$, together with a collection $C$ of finite non-empty subsets of $V$ such
that:
\begin{enumerate}
\item For any $v \in V$, the one-element set $\{v\}$ is in $C$;
\item If $\sigma \in C$ and $\sigma' \subseteq \sigma$, then $\sigma' \in C$.
\end{enumerate}

The elements of $V$ are called {\em vertices}, and the elements of $C$ are called {\em simplices}. We usually
drop $V$ from the notation, and refer to the simplicial complex as $C$.


A subset of a simplex is called a {\em face} of that simplex.

A {\em subcomplex} of $C$ is a subset of $C$ that is also a simplicial complex.

The {\em dimension} of a simplex $\sigma \in C$ is its cardinality minus one. The $k$-skeleton of a complex $C$,
denoted $\Skel^k C$, is the subcomplex formed of all simplices of $C$ of dimension $k$ or less.

A simplicial complex $C$ is called {\em pure} of dimension $n$ if $C$ has no simplices of dimension $> n$, and
every $k$-dimensional simplex of $C$ (for $k < n$) is a face of an $n$-dimensional simplex of $C$.


Let $A$ and $B$ be simplicial complexes. A map $f: A \to B$ is called {\em simplicial} if it is induced by a map
on vertices; that is, $f$ maps vertices to vertices, and for any $\sigma \in A$, we have $$ f(\sigma) =
\bigcup_{v \in \sigma} f(\{v\}).$$ A simplicial map $f$ is called {\em non-collapsing} (or {\em
dimension-preserving}) if $\dim f(\sigma) = \dim \sigma$ for all $\sigma \in A$.

Any simplicial complex $C$ has an associated {\em geometric realization} $|C|$, defined as follows: Let $V$ be
the set of vertices in $C$. As a set, we let $C$ be the subset of $[0,1]^V = \{\alpha : V \to [0,1]\}$
consisting of all functions $\alpha$ such that $\{ v \in V \mid \alpha(v) > 0 \} \in C$ and $\sum_{v \in V}
\alpha(v) = 1$.
For each $\sigma \in C$, we set $|\sigma| = \{ \alpha \in |C| \mid \alpha(v) \neq 0 \Rightarrow v \in \sigma
\}.$ Each $|\sigma|$ is in one-to-one correspondence with a subset of $\R^n$ of the form $\{(x_1, \dots, x_n) \in
[0,1]^n \mid \sum x_i = 1\}.$ We
put a metric on $|C|$ by $d(\alpha, \beta) = \sum_{v \in V} |\alpha(v) - \beta(v)|.$ 


A non-empty complex $C$ is called {\em $k$-connected} if, for each $m\leq k$, any continuous map of the
$m$-sphere into $|C|$ can be extended to a continuous map over the $(m+1)$-disk.

A {\em subdivision} of a simplicial complex $C$ is a simplicial complex $C'$ such that:
\begin{enumerate}
\item The vertices of $C'$ are points of $|C|$.

\item For any $\sigma' \in C'$, there exists $\sigma \in C$ such that $\sigma' \subset |\sigma|$.

\item The piecewise linear map $|C'| \to |C|$ mapping each vertex of $C'$ to the corresponding point of $C$ is a
homeomorphism.
\end{enumerate}



\myparagraph{Chromatic complexes.}
We now turn to the chromatic complexes used in distributed computing, and recall some notions from \cite{HS99}.

Fix $n \geq 0$. The {\em standard $n$-simplex} $\s$ has $n+1$ vertices, in one-to-one correspondence with $n+1$
{\em colors} $0, 1, \dots, n$. A face $\t$ of $\s$ is specified by a collection of vertices from $\{0, \dots,
n\}$. We view $\s$ as a complex, with its simplices being all possible faces $\t$.

A {\em chromatic complex} is a simplicial complex $C$ together with a non-collapsing simplicial map $\chi: C \to
\s$. Note that $C$ can have dimension at most $n$. We usually drop $\chi$ from the notation. We write $\chi(C)$
for the union of $\chi(v)$ over all vertices $v \in C$. Note that if $C' \subseteq C$ is a subcomplex of a
chromatic complex, it inherits a chromatic structure by restriction.

In particular, the standard $n$-simplex $\s$ is a chromatic complex, with $\chi$ being the identity.

Every chromatic complex $C$ has a {\em standard chromatic subdivision} $\Chr C$. Let us first define $\Chr \s$
for the standard simplex $\s$. The vertices of $\Chr \s$ are pairs $(i, \t)$, where $i \in \{0,1 ,\dots, n\}$
and $\t$ is a face of $\s$ containing $i$. We let $\chi(i, \t) = i$. Further, $\Chr \s$ is characterized by its
$n$-simplices; these are the $(n+1)$-tuples $((0,\t_0), \dots, (n, \t_n))$ such that:
\begin{enumerate}[(a)]
\item For all $\t_i$ and $\t_j$, one is a face of the other;
\item If $j \in \t_i$, then $\t_j \subseteq \t_i$. 
\end{enumerate} 
The geometric realization of $\s$ can be taken to be the set $\{\x=(x_0, \dots, x_n) \in [0,1]^{n+1} \mid \sum
x_i = 1\},$ with the vertex $i$ corresponding to the point $\x^i$ with $i$ coordinates $1$ and the other
coordinates $0$. Then, we can identify a vertex $(i, \t)$ of $\Chr \s$ with the point
\[
\frac{1}{2k-1} \x_i + \frac{2}{2k-1} \Bigl( \sum_{\{j \in \t \mid j \neq i\}} \x_j \Bigr) \ \in |\s| \subset
\R^{n+1},
\]
where $k$ is the cardinality of $\t$. 
Thus, $\Chr \s$ becomes a subdivision of $\s$ and the geometric realizations are identical: $|\s|=|\Chr \s|$. 

Next, given a chromatic complex $C$, we let $\Chr C$ be the subdivision of $C$ obtained by replacing each
simplex in $C$ with its chromatic subdivision. Thus, the vertices of $\Chr C$ are pairs $(p, \sigma)$, where $p$
is a vertex of $C$ and $\sigma$ is a simplex of $C$ containing $p$. If we iterate this process $m$ times we
obtain the $m^{th}$ chromatic subdivision, $\Chr^m C$.

Let $A$ and $B$ be chromatic complexes. A simplicial map $f: A \to B$ is called a {\em chromatic map} if for all
vertices $v \in A$, we have $\chi(v) = \chi(f(v))$. Note that a chromatic map is automatically non-collapsing. A
chromatic map has chromatic subdivisions $\Chr^m f: \Chr^m A \to \Chr^m B$. Under the identifications of
topological spaces $|A| \cong |\Chr^m A|, |B| \cong |\Chr^m B|,$ the continuous maps $|f|$ and $|\Chr^m f|$ are
identical.


\section{Omitted proofs}

\label{app:Proof_KprocSim}

{
\renewcommand{\thetheorem}{\ref{lem:KprocSim}}
\begin{lemma}
Algorithm~\ref{Alg:K-conc} provides a non-blocking simulation of a
$k$-process read-write shared-memory system in the $k$-set consensus
model.
Moreover, if there are $\ell<k$ active processes, then 
one of the first $\ell$ simulated processes is guaranteed to make progress.
\end{lemma}
\addtocounter{theorem}{-1}
}

\begin{proof}
We derive correctness of the simulation in Algorithm~\ref{Alg:K-conc} from
the following three properties: (1) every
simulated process follows a unique sequence of operations,
(2)~simulated snapshots and updates are linearizable, and (3)~at least one simulated process, with an
associated identifier lower or equal to the number of \emph{active}
processes, takes an infinite number of steps.

Let us first prove two useful simple claims:
\begin{enumerate}
\item \textbf{The write counter of a simulated process never
  decreases:} A write counter can only be modified  by incrementing it
  in line~\ref{Alg:Kc:WCincr} or by adopting it from the result of a
  commit-adopt operation in lines~\ref{Alg:Kc:1CA}
  or~\ref{Alg:Kc:OtherCA}. Moreover, the write counter and the
  associated value are only updated when the write counter obtained
  from a commit-adopt object is not smaller than the current one.
\item \textbf{If a process has a write counter equal to $c>0$, then at
  least one write/snapshot operation has been validated, i.e., a
  simulator passed the test in line~\ref{Alg:Kc:Commit}, for every
  write counter $c'$, $0\leq c'<n$:} Let us prove this property by
  induction. Initially, all write counters are set to $0$ and, thus,
  the property is trivially verified. Assume that the property holds
  for value $c$,
  and consider a state in which some process has a write counter equal to
  $c+1$. Consider the first time a process updates its write counter
  to $c+1$, it can only be the result of 
  line~\ref{Alg:Kc:WCincr}, as adopting the write counter from another
  process would result in a contradiction (this process must have updated it to $c+1$ first).
\end{enumerate}
Using these claims, let us show our three required properties:
\begin{itemize}
\item \textbf{For every couple $(s,c)$, where $s$ is a simulated
  process and $c$ is a write
  counter, all validated writes are identical.} According to the
  structure of 
  Algorithm~\ref{Alg:K-conc}, an operation is
  \emph{validated} if and only if it is returned from a call to a
  commit-adopt object with a \emph{commit} flag
  (line~\ref{Alg:Kc:Commit}). Consider the smallest round $r$ in which
  some process obtains a \emph{commit} flag with a write counter equal to $c$
  from a commit-adopt object associated to $s$. According to the
  specification of commit-adopt, every process obtains the same
  output value at round $r$ from this commit-adopt object, $(c,val)$, possibly with
  an \emph{adopt} flag. According to \textbf{Claim~2}, if no process
  \emph{validated} a write for $(s,c)$, then no write counter can be
  greater than $c$. Thus, every process replaces its current proposal
  value with the same couple $(c,val)$. 

By contradiction, assume that a process \emph{validates} a different write value for
$(s,c)$. By \textbf{Claim~1}, this process must have modified
its snapshot estimate without ever increasing its write
counter. The first process that modified the estimate could only do it
by adopting the proposition from another process, but no other process
changed yet its proposal value, which results in a contradiction.

\item \textbf{Simulated snapshots and updates are linearizable.} A
  simulated snapshot is the result of applying the auxiliary function
  $\textit{CurWrites}$ (lines~\ref{Alg:Kc:CurWrites1}--\ref{Alg:Kc:CurWrites2})
  on the snapshot result of $MEM$. The auxiliary function simply
  selects, for each simulated process, the write value with the
  greater write counter. By the previous property, all \emph{validated}
  writes updated to $MEM$ with the same counter are identical, and
  so, only the first completed write value may change the result of
  the simulated snapshot computation. Indeed, a \emph{validated} write
  can only be overwritten by a more recent write value, i.e., a write
  value associated with a greater write counter, as the write counter
  strictly increases during a \emph{validation} and can never decrease afterwards (\textbf{Claim~1}).

Therefore, for a given simulated process, the first write updating
$MEM$ for a given write counter will be selected in any later snapshot
until a write associated to a larger write counter is performed. 
Now we select the linearization point of a simulated write to be the
linearization point of the first
validated update (line~\ref{Alg:Kc:Write}) performed with it.
A simulated snapshot is
linearized to the linearization point of the corresponding
validated snapshot operation on $MEM$. 
The simulated snapshots and updates ordered by their
linearization points constitute a legal sequential history.
It is easy to see that a
validated snapshot operation, the one which served for the
simulated write computation, is linearized after the preceding
write operations of the simulated process. Indeed, a simulated snapshot is
computed only by the processes which \emph{validated} the preceding
write (line~\ref{Alg:Kc:Mainselect}) and only after the update was
made to $MEM$ (on line~\ref{Alg:Kc:Write}), thus, always after the
first update was made for the associated write counter.
\end{itemize}
Finally, we prove liveness of the simulation:

\begin{itemize}
\item \textbf{A simulated process with an identifier
  smaller or equal to the number of active processes takes infinitely 
  number of steps.} This result directly follows from the properties
  of $k$-simultaneous consensus objects. With the number of
  \emph{active} processes equal to $m$, processes are provided with an
  output value $(i,val)$ where $i<min(m,k)$ and every process with the
  same index obtains the same value. As processes first access the
  commit-adopt object associated with the obtained index $i$, some
  process must obtain a \emph{commit} flag for its index. Indeed, the
  first process to obtain an output for its first commit-adopt
  object in a round, may only have witness other processes with the
  same proposal. As a commit-adopt object must return a \textit{commit}
  flag when every proposed values are identical, this first process
  must obtain a \textit{commit} flag associated to its index. Thus, this
  process \emph{validates} the corresponding write value and increases its write counter. 

By \textbf{Claim~1}, write counters can never
decrease. Also, in every round, at least one simulator increases the
write counter associated with a simulated process with an identifier
smaller or equal to the number of \emph{active} processes,
$m$. Therefore, after an infinite number of rounds, the write counters
associated with such a simulated process have been incremented an
infinite number of times.  But there
are only finitely many of them, one per process and simulated process.
Thus, the write counter of one of them is incremented an infinite number of times. Using
\textbf{Claim~2}, an infinitely incremented write counter implies a
simulated process taking an infinite number  of simulation steps. 
\end{itemize}
\end{proof}

{
\renewcommand{\thetheorem}{\ref{lem:KconsSim}}
\begin{lemma}
In $\RK^*$, Algorithm~\ref{Alg:K-cons} provides a non-blocking simulation of any shared memory algorithm with access to $k$-set-agreement objects.
\end{lemma}
\addtocounter{theorem}{-1}
}

\begin{proof}
  We prove correctness of Algorithm~\ref{Alg:K-cons} in three steps:
  (1)~The read-write simulation is safe, i.e., the write-snapshot operations
are atomic; (2)~the $k$-set-agreement algorithm is safe, i.e., processes decide on
at most $k$ distinct proposed values for the same agreement object;
and (3)~the simulation is non-blocking, i.e., there is a
non-terminated process which completes an infinite number of simulated operations. 
\begin{itemize}
\item \textbf{The read-write simulation is safe}: The structure of the
  simulation is taken from an analogous simulation
  in~\cite{GR10-opodis}, therefore, this part of the proof is also
  directly inspired from the one in~\cite{GR10-opodis}. 

In Algorithm~\ref{Alg:K-cons}, memory operations are reduced to a
single write/snapshot operation. It is easy to see that any read-write
algorithm can be executed in this way by re-writing the same value
again to discard a write or by ignoring all or part of the snapshot
result to discard read operations. Nevertheless, even if a single
write/snapshot operation is provided, it is not an \emph{immediate}
snapshot, i.e., the write and the snapshot operations cannot be
linearized together by \emph{batches}. We will show that the set of
new written values returned in some simulated snapshot during round $r$,
$Wproc(snap(r))$, and the set of processes returning this snapshot,
$Rproc(snap(r))$, can be linearized firstly according to the
associated round number and secondly, for operations in the same
round, by linearizing the write operations before the read operations.

\begin{itemize}
\item \emph{Claim 1:} \textbf{Write counters can only increase.}
  Indeed, a write counter can only be modified by adopting a strictly
  greater value from another process (on
  lines~\ref{Alg:K-cons:l:UpdateWriteMin}--\ref{Alg:K-cons:l:UpdateWriteMax})
  or directly  incremented~\ref{Alg:K-cons:l:Increment}.
\item \emph{Claim 2:} \textbf{The sum of the write counters of an
  undecided process is greater or equal to the round number.} Let us
  show this claim with a trivial induction on the round number. The
  property is true for the first round, $r=1$, as the write counter
  associated to the processes own $ids$ is initially set to $1$ while
  the rest is set to $0$. Assume the property is true at round
  $r$. Then either the sum result is strictly greater than the round
  number and thus the property is true for round $r+1$, as write
  counters can only increase (\textbf{Claim 1}), or otherwise the sum
  is equal to the round number. In the later case the equality test
  made on line~\ref{Alg:K-cons:l:testCount} is verified and if the
  process is still \emph{undecided} it does not pass the test on
  line~\ref{Alg:K-cons:l:Terminate} and thus it increments its own
  write counter on line~\ref{Alg:K-cons:l:Increment} and the sum
  becomes, and stays (\textbf{Claim 1}), greater or equal to $r+1$.
\item \emph{Claim 3:} \textbf{There is a unique write value associated
  to each process and write counter, except possibly during the
  execution of line~\ref{Alg:K-cons:l:UpdateWriteMax} and
  lines~\ref{Alg:K-cons:l:NewOpMin}--\ref{Alg:K-cons:l:NewOpMax}.}
  This is a simple observation that the adoption of a write value is
  always made with its associated write counter on a memory position
  associated to the corresponding process (see
  lines~\ref{Alg:K-cons:l:UpdateWriteMin}--\ref{Alg:K-cons:l:UpdateWriteMax}). Otherwise,
  a write value can only be modified on
  line~\ref{Alg:K-cons:l:NewOpMax} by selecting a new write value but
  this is done only once the associated write counter has been
  incremented on line~\ref{Alg:K-cons:l:Increment}, a write counter
  that cannot have been used before, as write counters only increase
  (\textbf{Claim 1}).
\item \emph{Claim 4:} \textbf{A unique snapshot result can be returned
  per round.} According to the containment property of snapshots
  operations, the \IStwo output results of a call to an iteration of
  $\RK$ can be fully ordered by inclusion, i.e., $S_1\subset
  ... \subset S_m$. As the greatest write counter observed in these
  sets are adopted (see
  lines~\ref{Alg:K-cons:l:UpdateWriteMin}--\ref{Alg:K-cons:l:UpdateWriteMax}),
  a process with a larger \IStwo output obtains larger or equal write
  counters. Therefore the sum of the write counters of a process with
  a larger \IStwo output is equal to the sum of the write counters of
  a process with a smaller \IStwo output if and only if the write
  counters are equal for every corresponding process. But if processes
  obtain the same sum of write counters, the set of selected write
  values are identical (\textbf{Claim 3}). As in any round, a snapshot
  is returned (at line~\ref{Alg:K-cons:l:ValidateWrite}) only if the
  sum of write counters equals the round number (test on
  line~\ref{Alg:K-cons:l:testCount}), and so all returned snapshot are
  identical,  if any.
\item \emph{Claim 5:} \textbf{A write value can only be replaced with
  a more recent value.} This is a direct corollary of \textbf{Claims 1
  and 3}. Indeed, a write value associated with a given process can
  only replaced by a write value associated to the same process but
  with a greater associated write counter. Yet,	 a more recent write
  value is always associated with a greater write counter than all
  previously used ones for the corresponding process.
\end{itemize}

We therefore are provided with a unique sequence of validated
snapshots according to rounds, $snap(0)$, \dots, $snap(r)$,
$snap(r+1)$, \dots (\textbf{Claim 4}). We can thus define without
ambiguity $Wproc(snap(r))$ and $Rproc(snap(r))$, respectively the set
of firstly observed write values and the set of processes which
returned $snap(r)$ during round $r$. Moreover, according to
\textbf{Claim 5}, a write value can only be replaced by a more recent
write value, thus a snapshot returned on a later round can only
contain identical or more recent write values. It is obvious that
write operations can be safely linearized between the last snapshot
which observed an older write value and the first snapshot to return
the write value or a more recent one (as processes alternate write and
snapshot operations write operations can be linearized in any order in
between two snapshots, i.e., there can be at most one write value per
process between two consecutively linearized snapshots). Therefore the
following is a valid linearization ordering (we leave out the trivial
verification that the linearization order respects operations local ordering):

\begin{small}\[
\forall (r,r'),r<r', Lin(Rproc(snap(r)))<Lin(Wproc(snap(r')))<Lin(Rproc(snap(r'))){}.
\]
\end{small}

\item \textbf{The $k$-set-agreement is safe:} The safety of an agreement operation relies on two properties, \emph{validity} and \emph{agreement}. We also show an intermediary result later re-used for the liveness property, \emph{leader visibility}:
\begin{itemize}
\item \emph{Validity:} \textbf{Every decided value is the input proposal of some process:} 
The estimated decision value in $\textit{ConsHistory}_i$ is always either initialized to the process own proposal (on line~\ref{Alg:K-cons:l:AdoptFirst}), or to a value previously adopted from a \emph{leader} (on line~\ref{Alg:K-cons:l:AdoptLeader}). The estimated decisions values are then never dissociated from their corresponding $k$-set-agreement object identifier, $\textit{ConsId}_i$, in the local memory object $\textit{ConsHistory}_i$. It can further only be replaced in $\textit{ConsHistory}_i$ with the current initialized decision estimate, with the same object identifier, from another process (on line~\ref{Alg:K-cons:l:AdoptLeader}). Thus deciding on the estimated decision value in $\textit{ConsHistory}_i$ (at line~\ref{Alg:K-cons:l:Decide}) always decides on some process input proposal.
\item \emph{leader visibility:} \textbf{For every round, there is an undecided process, with an \ISone output containing at most $k$ inputs from undecided processes, contained in each undecided process $\RK$ output:} This property is directly derived from the definition of $\RK$ restricting the \emph{contention sets} sizes. Indeed, consider the smallest output obtained by an undecided process from an iteration of $\RK$, denoted $S_{min}$. Such a smallest view exists according to the snapshots \emph{containment} property. By the \emph{containment} property also, if $S_{min}$ contains an \ISone output with at most $k$ elements from undecided processes, then every $\RK$ output contains it. Moreover, by the \emph{self-inclusion} property, $S_{min}$ contains the \ISone output of an undecided process. 

Now let us assume that $S_{min}$ contains only \ISone outputs from undecided processes containing strictly more than $k$ inputs from undecided processes, and let $P_{min}$ be one of such a set of undecided processes with inputs contained in an \ISone output. As $S_{min}$ is the smallest output of an undecided process, all processes in $P_{min}$ observed in their $\RK$ output every $\RK$ input from processes in $P_{min}$. This implies that $P_{min}$ forms a \emph{contention set}, which is a contradiction with the assumption that $P_{min}$ includes more than $k$ processes.

\item \emph{Agreement:} \textbf{At most $k$ distinct values can be decided:} Consider the first round, $r$, at which some process completes an agreement operation for an agreement object (on line~\ref{Alg:K-cons:l:Decide}). A deciding process must have observed only round \emph{leaders}, i.e., processes with an \ISone output containing at most $k$ undecided processes inputs, with a decision estimate for the considered agreement object ($\textit{Leaders}_i=true$). Thus according to the \emph{leader visibility} property, there is a leader observed by every undecided process with a decision estimate for the considered agreement object. Therefore, during this round $r$, every process adopts the decision estimate of a \emph{leader} for the considered agreement object (at line~\ref{Alg:K-cons:l:AdoptLeader}). Thus every process has a decision estimate in its $ConsHistory$ object from one of the at most $k$ round \emph{leaders}. Moreover, even if processes start a call to this agreement object in a later round, they would discard their own input proposal value directly (see line~\ref{Alg:K-cons:l:PreventingAdopt}), and therefore the set of decision estimates can only decrease in later rounds.

It is easy to see by combining the \emph{self-inclusion} and \emph{containment} properties that at most $k$ processes can obtain a snapshot output containing at most $k$ inputs. Thus there can be at most $k$ \emph{leaders} in a round, therefore, the set of decided values being bounded by the number of decision estimates adopted at round $r$, at most $k$ distinct values can be decided for a given agreement object.
\end{itemize} 
\item \textbf{The simulation is non-blocking:} A non-blocking simulation means that if \emph{undecided} processes complete an infinite number of steps, then there is an \emph{undecided} process which completes an infinite number of simulation operations. Showing this property is equivalent to show that there is no reachable configuration where there are \emph{undecided} processes never completing any simulation steps during infinitely many algorithm rounds (the algorithm does not include waiting statements or infinite loops). 

Let us assume we are in such a state. We will first show that the
write counter of an \emph{undecided} process is increased infinitely
many times. This simply results from the previously proven
\textbf{Claim 2} from the safety proof of the read write simulation, stating that the sum of write counters of an undecided process is always greater than the round number. Thus, as there is an undecided process completing an infinite number of simulation rounds, there is a write counter increased infinitely many often. This write counter can only be one of an \emph{undecided} process as a write counter is never increased after termination (see lines~\ref{Alg:K-cons:l:Terminate} and~\ref{Alg:K-cons:l:Increment}). Thus there is a process passing the test on line~\ref{Alg:K-cons:l:testCount} infinitely often.

If a process passes the test on line~\ref{Alg:K-cons:l:testCount}, then every \emph{undecided} process with a smaller $\RK$ output, thus a smaller write counter sum, must pass it as well. Let us consider the \emph{undecided} process, $p_{min}$, which passes this test infinitely often with the smallest $\RK$ output obtained by an \emph{undecided} process. This process must have a pending active agreement operation as otherwise it would complete its write operation, or \emph{terminate}, or select a new operation. Without loss of generality, we can take $p_{min}$ to be indefinitely often a \emph{leader} seen by every undecided process (see the \textbf{leader visibility} property). Therefore every process must eventually obtain a decision estimate for $p_{min}$ pending agreement object (at line~\ref{Alg:K-cons:l:AdoptLeader}), and thus $p_{min}$ eventually decides as it eventually only observe leaders with a decision estimate ($Leader=true$). This results in a contradiction, thus the simulation is non-blocking.
\end{itemize}
\end{proof}
\end{document}